\documentclass{article}
\usepackage[T2A]{fontenc}
\usepackage[utf8]{inputenc}
\usepackage[russian,english]{babel}

\usepackage{amsmath,amstext,amsthm,amssymb}

\newtheorem{theorem}{Theorem}

\theoremstyle{remark}
\newtheorem{definition}{Definition}
\newtheorem*{remark}{Remark}

\newcommand{\ia}{I}
\newcommand{\m}{m}
\newcommand{\N}{\mathbb N}
\newcommand{\Q}{\mathbb Q}
\newcommand{\R}{\mathbb R}

\newcommand{\KP}{\textit{K}}
\newcommand{\IP}{\textit{I}}

\newcommand{\pair}[1]{\langle #1\rangle}
\newcommand{\ave}{\mathbf E}
\newcommand{\lea}{\preccurlyeq}
\newcommand{\eqa}{\asymp}

\title{Proofs of conservation inequalities for Levin's notion of mutual information of 1974}
\author{Nikolay Vereshchagin\thanks{The article
  was prepared within the framework of the HSE University Basic Research Program and funded by the Russian Academic Excellence Project '5-100'. The author
  was in part funded by RFBR according to the research project № 19-01-00563.}\\
Moscow State University
and  HSE University, Russian Federation.}
\begin{document}
\maketitle

\begin{abstract} 
In this paper we consider  Levin's notion of mutual information in infinite 0-1-sequences, as  
defined  in [Leonid Levin.  Laws of Information Conservation (Nongrowth) and Aspects of the Foundation of Probability Theory. Problems of information transmission, vol. 10 (1974),
pp. 206--211]. The respective information conservation  inequalities were stated in that paper without proofs.
Later some proofs appeared in the literature, however  no proof of the 
probabilistic conservation  inequality has been published yet. 
In this paper we prove that inequality and 
for the sake of completeness we present also short proofs of other properties of the said notion.
\end{abstract}

\section{Mutual information in finite strings}

\begin{definition}\label{def1}
For binary strings $x,y\in\Xi$ the notion of mutual information is  defined by the formula
 $$
 \IP(x:y)=\KP(x)+\KP(y)-\KP(x,y).
 $$
 \end{definition}

We use the following notation:
\begin{itemize}
\item $\Xi$ denotes the set of all binary strings. 
 \item $K(x),K(x|y)$ denote Kolmogorov prefix  and Kolmogorov prefix  conditional complexity (for the definition and 
more details see~\cite{lv,suv}).
\item $\KP(x,y)$ denotes $\KP(\pair{x,y})$ where $(x,y)\mapsto \pair{x,y}$
stands for a fixed computable injective mapping from 
$\Xi\times\Xi$ to $\Xi$.
\end{itemize}

It is not hard to prove the following facts called 
\emph{Information Conservation Inequalities:}

                          \begin{theorem} (\cite{PPI})
                          \label{detnongrowthfin}
Let  $A$ be an algorithm. Then for all $x,y\in\Xi$ such that $A(x)$ is defined we have 
$\IP(A(x):y)\le \IP(x:y)+c_A$.
                          \end{theorem}

                          \begin{theorem} (\cite{PPI})
                          \label{probnongrowthfin}
Let  $P_x:\Xi\to\Q$, $ x\in\Xi$,  be a uniformly computable family of probability measures\footnote{For technical simplicity,
we assume that all measures in this paper take only rational values.} on $\Xi$ (there is
an algorithm $A$ that computes $P_x(z)$ from $x,z$). 
Then for all  $x,y$ we have 
$$
\ave_{P_x(z)}2^{\IP(\pair{x,z}:y)}\le 2^{\IP(x:y)+c_A}
$$ 
where $\ave_{P_x(z)}$ denotes the average with respect to $P_x(z)$. 
 \end{theorem}

Actually, this theorem is stated in~\cite{PPI} in terms of tests.

 \begin{definition}
 A $P$-test, where $P$ is a probability distribution on  some space $Z$,
 is a function from $Z$ to $\N$.
 A $P$-test $t(z)$ is \emph{expectation bounded} if  $\ave_{P(z)}2^{t(z)}\le1$.
 A $P$-test is \emph{probability bounded} if  $P$-probability of the 
 set of all outcomes $z\in Z$ with $t(z)\ge n$ is at most $2^{-n}$
 for all $n\in\N$.
 \end{definition}
By Markov inequality, every expectation bounded $P$-test is
probability bounded as well. Using the notion of a test, the conclusion of Theorem~\ref{probnongrowthfin} can be restated as follows:    \\          
{\it For all $x,y$ there is an expectation bounded $P_x$-test $t_{x,y,A}$ 
such that 
$$
\IP(\pair{x,z}:y)\le \IP(x:y)+t_{x,y,A}(z)+c_A.
$$ 
}
\begin{proof} Indeed, in one direction (from Theorem~\ref{probnongrowthfin} to this statement) we 
 let  $$t_{x,y,A}(z)=\IP(\pair{x,z}:y)-\IP(x:y)-c_A.$$
In the other direction we just note that   
$$
\ave_{P_x(z)}2^{\IP(x:y)+t_{x,y,A}(z)+c_A} =  2^{\IP(x:y)+c_A}       \ave_{P_x(z)}2^{t_{x,y,A}(z)}\le   2^{\IP(x:y)+c_A}      \cdot1,
$$
where the inequality holds by the definition of an expectation bounded test.
\end{proof}

\section{Mutual information in infinite binary sequences: the definition}

For infinite binary sequences $\alpha,\beta\in\Omega$, Levin 
defined in~\cite{PPI} the following notion of mutual information $\ia(\alpha:\beta)$ in $\alpha,\beta$:
\begin{definition} \label{def3}
                       $$
\ia(\alpha:\beta)=\log\sum_{x,y\in\Xi}2^{I(x:y)-K^\alpha(x)-K^\beta(y)}.
                       $$
\end{definition}

Here we use the following notation:

\begin{itemize}
\item  $\Omega$ denotes the set of all infinite binary sequences.
\item $K^\alpha(x)$ denotes the prefix complexity 
relativized by $\alpha\in\Omega$.
\item $\log x$ denotes base 2 logarithm. 
\end{itemize}

\begin{remark}\small
This definition can be motivated as follows.
A natural attempt to define a notion of mutual information of $\alpha,\beta$
would be to let 
$\ia_1(\alpha:\beta)=\sup I(\alpha_{1:n}:\beta_{1:n})$,
where $\alpha_{1:n}$ stands for the length-$n$ prefix of $\alpha$.
This definition is bad, since 
$\alpha_{1:n}$ may have more information than $\alpha$, as it identifies $n$.
By this reason $\ia_1(\alpha:\beta)$ is always infinite, even for
all zero sequences.

To avoid this drawback, we could subtract  $K^\alpha(n)$ or $K^\beta(n)$
(or both) from  $I(\alpha_{1:n}:\beta_{1:n})$, or to consider $I(\alpha_{1:n}:\beta_{1:n}|n)$
instead\footnote{$I(x:y|z)$ is defined as $K(x|z)+K(y|z)-K(\pair{x,y}|z)$}, but there is another criticism.
Why we compare only prefixes of the same length of $\alpha$ and
$\beta$? Assume, for example, that we have stretched $\alpha$ by inserting zeros 
in all odd positions. Then the information in $\alpha$ has not changed,
but after this change  $\alpha_{1:n}$ becomes   $\alpha_{1:2n}$.

To avoid both problems, it is natural to let   
\begin{align*}
\ia_3(\alpha:\beta)&=\sup \{I(\alpha_{1:n}:\beta_{1:m})-K^\alpha(n)-K^\beta(m)\mid n,m\in \N\}\\
&=\sup\{I(x:y)-K^\alpha(x)-K^\beta(y)\mid x \text{ is a prefix of }\alpha\text{ and }y \text{ is a prefix of }\beta\}.
\end{align*}
This definition is already very similar to Definition~\ref{def3}:
it is obtained from Definition~\ref{def3}
by replacing summation by supremum  
and reducing the range of $x,y$  to prefixes of $\alpha,\beta$.
Using summation instead of supremum 
is crucial  for Theorem~\ref{probnongrowth2} below.
(The advantage of summation over maximization is
that the former commutes with another summation while the latter does not.)
Keeping the full range of $x,y$ is important for Theorem~\ref{detnongrowth2} below. (End of remark.)
\end{remark}

It turns out that this definition generalizes Definition~\ref{def1} of mutual information 
for finite strings:
                          \begin{theorem}[\cite{PPI}]\label{th1}
Let $\tilde u$ stand for the infinite sequence
that starts with a the prefix code of $u$ \footnote{say, $0^{n}1u$, where $n$ denotes the length of $u$} and then all zeros. Then
$\ia(\tilde u:\tilde v)=I(u:v)+O(1)$.
                          \end{theorem}
   The notation $f(z)\le g(z)+O(1)$ means that there is a constant $c$ (depending on the choice 
of universal machines in the definition of the prefix complexity and conditional prefix complexity) such that $f(z)\le g(z)+c$ for all $z$.                       
The notation $f(z)= g(z)+O(1)$ means that         $f(z)\le g(z)+O(1)$    and $g(z)\le f(z)+O(1)$  simultaneously.            
                          
There is no  proof of this theorem in \cite{PPI}. However,
several its proofs appeared in other papers (e.g. \cite{Geiger,bbgs}). For the sake of completeness we present here yet another proof.
In the proof we use the following notions, notations  and results:
\begin{enumerate}
\item A function $f:\Xi\to\R_{\ge0}$ is called
\emph{lower semicomputable} if there is an enumerable 
set $S$ of pairs of the form (a positive rational number $r$, a binary word $x$) such that
$f(x)=\sup\{r\mid (r,x)\in S\}$ (where supremum of the empty set is defined as 0).
The notion of a lower semicomputable
 function $f:\Xi\times\Xi\to\R_{\ge0}$ is defined in a similar way:
 this time $S$ consists of triples and not pairs.
 Similarly, we define the notion of a lower semicomputable
 function $f:\Xi\times\Omega\to\R_{\ge0}$:
there is a family of  sets $S^\alpha$, $\alpha\in\Omega$, such that 
there is an oracle Turing machine that enumerates  $S^\alpha$ with oracle $\alpha$  (the same machine for all oracles $\alpha$)
and $f(x,\alpha)=\sup\{r\mid (r,x)\in S^\alpha\}$.
\item A function $f:\Xi\to\R_{\ge0}$ is called a \emph{semimeasure}
if $\sum_{x\in\Xi} f(x)\le 1$. 
A function $f:\Xi\times\Xi\to\R_{\ge0}$ [$f:\Xi\times\Omega\to\R_{\ge0}$] is called \emph{a family of semimeasures}
if $\sum_{x\in\Xi} f(x,y)\le 1$ for all $y$ [$\sum_{x\in\Xi} f(x,\omega)\le 1$ for all $\omega\in\Omega$] . 
\item Let $\m(x)$ and $m(x|y)$ denote $2^{-K(x)}$ and $2^{-K(x|y)}$, respectively.  By G\'acs--Levin theorem 
$m(x)$ [$m(x|y)$] is a largest up to a constant factor lower semicomputable semimeasure
[family of semimeasures]: for every   lower semicomputable semimeasure $f(x)$ [family of semimeasures $f(x,y)$]
there is $c_f$ such that $f(x)\le c_fm(x)$ for all $x$ [$f(x,y)\le c_f m(x|y)$ for all $x,y$].  The functions $\m(x)$ and $m(x|y)$ are called  the \emph{a priori probability} and the \emph{conditional a priori probability} on $\Xi$.
\item  Let $m^\alpha(x)$ denote 
$2^{-K^\alpha(x)}$. 
Likewise $m^\alpha(x)$ is a largest up to a constant factor lower semicomputable family of
semimeasures: 
for every  lower semicomputable family of semimeasures $f:\Xi\times\Omega\to \R_{\ge0}$ 
there is $c_f$ such that $f(x,\alpha)\le c_fm^\alpha(x)$ for all $x,\alpha$. 
\item $f(z)\lea g(z)$ means that there is a constant $c$ (depending on the choice 
of universal machines in the definition  of the prefix complexity and conditional prefix complexity) such that $f(z)\le c\cdot g(z)$ for all $z$.
 \item $f(z)\eqa g(z)$ means that $f(z)\lea g(z)$ and $g(z)\lea f(z)$.
\item $\sum_{x\in\Xi} m(x,y)\eqa m(y)$~\cite[Problem~101 on p. 97]{suv}) \label{p6}.
\item $m(x,z)m(y,z)\lea  m(z)m(x,y,z)$,   the ``basic inequality'' (for the proof see~\cite[Theorem 70 on p. 110]{suv}).\label{p7}
\end{enumerate}

\begin{proof}
(a) Using a priori probability, the definition 
of mutual information can be rewritten as follows:
       $$
\ia(\alpha:\beta)
=\log\sum_{x,y\in\Xi}\frac {m^\alpha(x)m^\beta(y)m(x,y)}{m(x)m(y)}.
                       $$
We have $m^{\tilde u}(x)\eqa m(x|u)$.
 Thus in one direction we have to show that 
$$
\sum_{x,y\in\Xi}\frac {m(x|u)
m(y|v)m(x,y)}{m(x)m(y)}\lea 
\frac {m(u,v)}{m(u)m(v)}.
$$
After moving $m(u)m(v)$ from the right hand side to the left hand side
this inequality becomes
$$
\sum_{x,y\in\Xi}\frac {m(x|u)m(u)
m(y|v)m(v)m(x,y)}{m(x)m(y)}\lea 
m(u,v)
$$
First notice that $m(x|u)m(u)\lea  m(x,u)$ (in the logarithmic form: $K(x,u)\le K(u)+K(x|u)+O(1)$, see, e.g.,~\cite{lv,suv}). 
Then use  the basic inequality  $m(x,u)m(x,y)\lea  m(x)m(x,y,u)$. 
Thus the left hand side of the sought inequality is at most
$$
\sum_{x,y\in\Xi}\frac {m(x,y,u)
m(y|v)m(v)}{m(y)}.
$$
Now the only term that includes $x$ is $m(x,y,u)$ and
$$ 
\sum_{x\in\Xi}m(x,y,u)\eqa m(y,u)
$$
by item~\ref{p6} above.
Thus we get
$$
\sum_{y\in\Xi}\frac {m(y,u)
m(y|v)m(v)}{m(y)}.
$$
In a similar way we can show that 
$$
\frac {m(y,u)
m(y|v)m(v)}{m(y)}\lea  \frac {m(y,u)
m(y,v)}{m(y)}\lea  m(y,u,v).
$$
And again  by item~\ref{p6} we have 
$$
\sum_{y\in\Xi}m(y,u,v)\eqa m(u,v).
$$

(b) The other direction is trivial, since for $\alpha=\tilde u$ and $\beta=\tilde v$ the sum in Definition~\ref{def3}
has the term
$$
2^{I(u:v)-K(u|u)-K(v|v)+O(1)}=2^{I(u:v)+O(1)}.\qed
$$
\renewcommand{\qed}{}
\end{proof}

An interesting case is when one sequence is infinite and the other one is finite.
In this case it is usual to define the notion of mutual information by the formula $I(u:\beta)= K(u)-K^\beta(u)$,
where $ u\in\Xi$ and $\beta\in\Omega$.
It turns out that this definition differs from Levin's one:
    \begin{theorem}\label{th4}
    Assume that $ u\in\Xi$ and $\beta\in\Omega$. Then the following are true:
(a) $\ia(\tilde u:\beta)\ge K(u)-K^\beta(u)-O(1)$. (b) The converse inequality is false in general case. (c)  The converse inequality holds, 
if $0'$ (the Halting problem) Turing reduces to $\beta$:
 $\ia(\tilde u:\beta)\le K(u)-K^\beta(u)+O(1)$. (d) In any case $\ia(\tilde u:\beta)\le K(u)+O(1)$.

                          \end{theorem}
\begin{proof} (a)
In Definition~\ref{def3}, let $\alpha=\tilde u$ and consider the term in the sum where $x=y=u$.

(b) Let $\beta=\tilde v$ and $u=K(v)$, where $v\in\Xi$. Then the left hand side  of the inequality evaluates to
$$
I(K(v):v)+O(1)=K(K(v))+K(v)-K(v,K(v))+O(1)=K(K(v))+O(1)
$$ 
(the last equality holds by G\'acs theorem~\cite{gacs}).
And its right hand side  equals 
$K(K(v))-K(K(v)|v)+O(1)$. And by~\cite{gacs,bs} we know that 
$K(K(v)|v)$ can be $\Omega(\log n)$ for strings $v$ of length $n$.  

(c) As in Theorem~\ref{th1}(a), we can show
that
$$
2^{I(\tilde u:\beta)-K(u)}\lea
\sum_{y\in\Xi}\frac {m(y,u)m^\beta(y)}{m(y)}.
$$
The function $f(u)=\sum_{y\in\Xi}\frac {m(y,u)m^\beta(y)}{m(y)}$ from the
right hand side of this inequality is lower semicomputable
relative to $0'$ and hence relative to $\beta$. 
And $f(u)$ is a semimeasure (up to a constant factor), since  
$$
\sum_{u\in\Xi}\sum_{y\in\Xi}\frac {m(y,u)m^\beta(y)}{m(y)}=
\sum_{y\in\Xi}\sum_{u\in\Xi}\frac {m(y,u)m^\beta(y)}{m(y)}\eqa
\sum_{y\in\Xi}\frac {m(y)m^\beta(y)}{m(y)}\le1
$$
(in the equality we have changed the order of summation). 
Hence for some $c$ depending on $\beta$ we have 
$$
\sum_{y\in\Xi}\frac {m(y,u)m^\beta(y)}{m(y)}\le c\cdot m^\beta(u).
$$
Moreover, in a standard way we can construct a
lower  semicomputable family of semimeasures $g(u,\gamma)$ such that
$g(u,\beta)=f(\beta)$. Since 
$g(u,\gamma)\le c\cdot m^\gamma(u)$ for a constant $c$ that does not depend on $\beta$,
we have 
$2^{I(\tilde u:\beta)-K(u)}\lea 2^{-K^\beta(u)}$ and we are done.

(d) In item (c) we have seen that 
$$
2^{I(\tilde u:\beta)-K(u)}\lea
\sum_{y\in\Xi}\frac {m(y,u)m^\beta(y)}{m(y)}.
$$
Since $m(y,u)\lea m(y)$ the right hand side of this inequality is bounded by a constant.
\end{proof}
 
 An interesting problem is to find a simpler expression for $I(\tilde u:\beta)$
( where $ u\in\Xi$ and $\beta\in\Omega$) than the expression in Definition~\ref{def3}.
 
\section{Information Conservation Inequalities for infinite sequences}

The quantity $I(\alpha:\beta)$ satisfies \emph{Information Conservation Inequalities},
Theorems~\ref{detnongrowth2} and~\ref{probnongrowth2} below,
which are similar to  Theorems~\ref{detnongrowthfin} and~\ref{probnongrowthfin}. 
Theorems~\ref{detnongrowth2} and~\ref{probnongrowth2} were  published in \cite{PPI} without proofs.

                          \begin{theorem}[\cite{PPI}]
                         \label{detnongrowth2}
Assume that  $A$ is an algorithmic operator.
Then  $$\ia(A(\alpha):\beta)\le \ia(\alpha:\beta)+K(A)+O(1)$$
for all $\alpha,\beta$ such that $A(\alpha)$ is infinite.
                          \end{theorem}
By algorithmic operator here we mean an oracle Turing machine that has a write-only one-way output tape with binary tape alphabet.
By $K(A)$ we denote the prefix complexity of the program of that machine.
Such a machine with an oracle $\alpha$
prints a finite or infinite binary sequence on its output tape. We denote that sequence by  $A(\alpha)$. 
If  $A(\alpha)$ is finite, then  $A(\alpha)$ may have more information than $\alpha$ and   
$A$ together and the statement of the theorem may be false in that case.
For example, let $\alpha$ consist of a prefix code of (the binary expansion of) a natural number $n$
followed by zeroes  and let  $A$ extract $n$ from $\alpha$, print the prefix code of $n$ and then
run all programs of length at most $n$
dovetail style and append to the output the prefix code of $x$ whenever one of those programs halts with a result $x$.
Then  $A(\alpha)$ contains $n$ and a list of all strings of
complexity  at most $n$.
Hence $K(A(\alpha))\ge n-O(1)$, while $K(A)=O(1)$ and $\alpha$ has $O(\log n)$ bits of
information about any sequence (see Theorem~\ref{th7}(a) below).

\begin{proof}                                           
The inequality follows immediately from the inequality
$$
K^\alpha(x)\le K^{A(\alpha)}(x)+K(A)+O(1),
$$ 
which holds for all $\alpha,x,A$ such that 
$A(\alpha)$ is infinite (otherwise $ K^{A(\alpha)}(x)$ is undefined).
To prove this inequality, note
that, relative to $\alpha$, the string $x$ can be identified by the shortest program for $A$ 
followed by the shortest  program $p$ of $x$ relative to $A(\alpha)$.
Given such description, we first find the program of $A$
and then we run the (fixed) oracle Turing machine that,
given $A(\alpha)$ as oracle and $p$ as input prints $x$. 
When that machine queries $i$th bit of its oracle, we run the oracle Turing machine $A$ that prints $A(\alpha)$
until  its $i$th bit appears. We then use that bit as the oracle's answer to the query. If  the length of $A(\alpha)$ is less than $i$, then 
we wait  forever and do not print any result. 
\end{proof}
  
By this theorem, Definition~\ref{def3} is encoding invariant. This means that 
for any sequences $\alpha,\beta\in\Omega$
that can be obtained from each other by algorithmic operators we have 
$I(\alpha:\gamma)-I(\beta:\gamma)=O(1)$ for every $\gamma\in\Omega$.
 More specifically,  $$
 |I(\alpha:\gamma)-I(\beta:\gamma)|\le\max\{K(A),K(B)\}+O(1),
 $$
 if $A(\alpha)=\beta$ and $B(\beta)=\alpha$. In particular,
 in the next theorem we will need to encode 
 pairs of sequences $\rho,\omega\in\Omega$ by one sequence denoted by  $\pair{\rho,\omega}$.
By the encoding invariance, the value 
 $I(\pair{\rho,\omega}:\alpha)$ changes by at most $O(1)$, if we switch to another encoding.
 For instance, we can let $\pair{\rho,\omega}$ be equal to
the sequence whose odd bits are those of $\rho$ and even bits are those of $\omega$.
  
 The second ``probabilistic'' Information Conservation inequality is the following):
  
                          \begin{theorem}[we cite \cite{PPI}]
                            Let 
$P$ be an arbitrary probability distribution on $\Omega$, and let  $\rho\in\Omega$ be a sequence with
respect to which $P$ is computable
(such a sequence always exists). Then for all $\alpha\in\Omega$
there is a probability bounded $P$-test $t_{\alpha}$ such 
that 
$$
\IP(\pair{\rho,\omega}:\alpha)\le \IP(\rho:\alpha)+t_{\alpha}(\omega)+O(1)
$$ 
for all $\omega\in\Omega$. 
                          \end{theorem}
This statement needs a clarification what the constant $O(1)$ depends on, because it must depend on $P$.
Here is such a clarification:
\addtocounter{theorem}{-1}

                          \begin{theorem}[a clarification]
\label{probnongrowth2}
Assume that a family $P_\rho$, $\rho\in\Omega$, of probability distributions on $\Omega$ is fixed.
Assume that there is a Turing machine $T$ that for all $\rho$
computes  $P_\rho$ having oracle  access to $\rho$ \footnote{This means that 
the machine computes the probability of the cylinder $\Omega_x=\{\alpha\in\Omega\mid 
x \text{ is a prefix of }\alpha\}$ for any given  $x\in\Xi$.}. Then for all $\alpha,\rho\in\Omega$
there is a probability (and even expectation) bounded $P_\rho$-test $t_{\alpha,\rho,T}$ such 
that 
$$
\IP(\pair{\rho,\omega}:\alpha)\le \IP(\rho:\alpha)+t_{\alpha,\rho,T}(\omega)+c_{T}
$$ 
for all $\omega\in\Omega$, where $c_T$ does not depend on $\rho,\alpha,\omega$. 
                          \end{theorem}

\begin{proof}
Here again the definition 
of mutual information in terms of a priori probability is useful:
       $$
2^{\ia(\alpha:\beta)}
=\sum_{x,y\in\Xi}\frac {m^\alpha(x)m^\beta(y)m(x,y)}{m(x)m(y)}.
                       $$
Basically the statement of the theorem means
that 
 $$
\ave_{P_\rho(\omega)}\sum_{x,y\in\Xi}\frac {m^{\rho, \omega}(x)
m^{\alpha}(y)m(x,y)}{m(x)m(y)}
                       $$
 is not larger than the sum
  $$
\sum_{x,y\in\Xi}\frac {m^{\rho}(x)
m^{\alpha}(y)m(x,y)}{m(x)m(y)}
                       $$
up to a multiplicative  constant $d_T$ depending on $T$ (and then we let $c_T=\log d_T$).
These sums differ only in the first term in the numerators.
Hence it suffices to show that 
the $P_\rho$-expectation of  $m^{\rho, \omega}(x)$
is at most $d_T\cdot m^{\rho}(x)$:
$$
\int m^{\rho, \omega}(x)\,d P_\rho(\omega)\le d_T\cdot m^{\rho}(x).
$$
Recall that $m^{\rho}(x)$ is a largest (up to a multiplicative constant) lower semicomputable family of semimeasures.
Hence it suffices to show that
$$
f(x,\rho)=\int m^{\rho, \omega}(x)\,d P_\rho(\omega)
$$
is a lower semicomputable family of semimeasures.
Since  $P_\rho$ can be computed by a Turing machine with oracle $\rho$, the function $f(x,\rho)$       
is lower semicomputable.
Thus to finish the proof, it suffices to show that
this function is a family of semimeasures:
  $$
\sum_{x\in\Xi} \int m^{\rho, \omega}(x)\,d P_\rho(\omega)\le1.
$$          
We can swap summation and integration.
 Notice that $\sum_{x} m^{\rho, \omega}(x)\le1$ for all $\rho,\omega\in\Omega$ and hence
 $$
 \int \Big (\sum_{x} m^{\rho, \omega}(x)\Big)\,d P_\rho(\omega)\le
  \int   1\cdot \,d P_\rho(\omega)=1. \qed
$$
\renewcommand{\qed}{}
\end{proof}

\section{Mutual information and computability}

There is another good property of Definition~\ref{def3}:
if $\alpha$ is computable, then
 $I(\alpha:\beta)=O(1)$ for all $\beta$. 
 However, the converse is false:
 there are non-computable sequences $\alpha$ with 
 $I(\alpha:\beta)=O(1)$ for all $\beta$.
 
 \begin{theorem}\label{th7}
(a) If $\alpha\in\Omega$ is computed by a program $p$,
 then  $I(\alpha:\beta)\le K(p)+O(1)$ for all $\beta$. 
 (b) There is a non-computable $\alpha\in\Omega$ such that  $I(\alpha:\beta)\le c_\alpha<\infty$ for all $\beta$. 
 \end{theorem}
 \begin{proof}
  (a) The statement follows from the inequalities 
 $$
 K(x)\le K(x|p)+K(p)+O(1)\le K^\alpha(x) +K(p)+O(1).
 $$
 These inequalities imply that
 the sum in Definition~\ref{def3}
 is at most
 \begin{align*}
 \sum_{x,y\in\Xi}2^{I(x:y)-K(x)+K(p)-K^\beta(y)}=\sum_{x,y}2^{K(y)-K(x,y)+K(p)-K^\beta(y)}\\
 =\sum_y\Big( 2^{K(y)+K(p)-K^\beta(y)} \sum_{x}m(x,y)\Big).
\end{align*}
Since $\sum_{x}m(x,y)\eqa m(y)=2^{-K(y)}$,
this is at most
$$ 
\sum_{y\in\Xi} 2^{K(p)-K^\beta(y)}=2^{K(p)}\sum_{y\in\Xi} 2^{-K^\beta(y)}\le 2^{K(p)}\cdot1.
                       $$

(b) There are non-computable sequences $\alpha$ with  $K^\alpha(x)\ge K(x)-c$
for some  $c$ that does not depend on $x$ (see e.g.  \cite{dh}). For such sequences 
we can repeat the arguments from item (a).
\end{proof}

\end{document}